%% file: camera_ready.tex
\documentclass[10pt, letter, a4paper]{./lib/IEEEtran}
\usepackage{amsmath, amssymb}
\usepackage{mathtools}
\usepackage{array}
\usepackage{cite}
\usepackage{algpseudocode,algorithm,algpseudocode}
\usepackage{float}
\usepackage{tabularx}
\usepackage{mathrsfs} 
\usepackage{tikz}
\usetikzlibrary{chains,patterns, arrows,shapes,positioning,arrows,decorations.markings,calc}
\usepackage{pgfplots}
\pgfplotsset{compat=newest}
\usepgfplotslibrary{groupplots}
\usepgfplotslibrary{fillbetween}
\usepackage{gensymb}
\usepackage{soul}	
\usepackage{xcolor}
\usepackage[shortcuts,acronym]{glossaries}
\makeglossaries
\usepackage{fixmath}
\usepackage{amssymb}
\usepackage{bm}
\usepackage{bbm}
\usepackage{pbox}
\usepackage{subfig}
\usepackage{multirow}
\usepackage{tabularx}
\pgfplotsset{compat=newest,
	/pgfplots/ybar legend/.style={
		/pgfplots/legend image code/.code={%
			\draw[##1,/tikz/.cd,bar width=3pt,yshift=-0.2em,bar shift=0pt]
			plot coordinates {(0cm,0.8em)};},
	},}
\newcommand\myeqa{\stackrel{\mathclap{\mbox{($a$)}}}{=}}
\newcommand\myeqb{\stackrel{\mathclap{\mbox{\scriptsize{($\eta=4$)}}}}{=}}

\usepackage{amsthm}
\theoremstyle{plain}
\newtheorem{theorem}{Theorem}

\usepackage{caption}
\makeatletter
\def\therule{\makebox[\algorithmicindent][l]{\hspace*{.5em}\vrule height .75\baselineskip depth .25\baselineskip}}%

\newtoks\therules
\therules={}
\def\appendto#1#2{\expandafter#1\expandafter{\the#1#2}}
\def\gobblefirst#1{
	#1\expandafter\expandafter\expandafter{\expandafter\@gobble\the#1}}%
\def\LState{\State\unskip\the\therules}
\def\pushindent{\appendto\therules\therule}
\def\popindent{\gobblefirst\therules}
\def\printindent{\unskip\the\therules}
\def\printandpush{\printindent\pushindent}
\def\popandprint{\popindent\printindent}

\algdef{SE}[WHILE]{While}{EndWhile}[1]
{\printandpush\algorithmicwhile\ #1\ \algorithmicdo}
{\popandprint\algorithmicend\ \algorithmicwhile}%
\algdef{SE}[FOR]{For}{EndFor}[1]
{\printandpush\algorithmicfor\ #1\ \algorithmicdo}
{\popandprint\algorithmicend\ \algorithmicfor}%
\algdef{S}[FOR]{ForAll}[1]
{\printindent\algorithmicforall\ #1\ \algorithmicdo}%
\algdef{SE}[LOOP]{Loop}{EndLoop}
{\printandpush\algorithmicloop}
{\popandprint\algorithmicend\ \algorithmicloop}%
\algdef{SE}[REPEAT]{Repeat}{Until}
{\printandpush\algorithmicrepeat}[1]
{\popandprint\algorithmicuntil\ #1}%
\algdef{SE}[IF]{If}{EndIf}[1]
{\printandpush\algorithmicif\ #1\ \algorithmicthen}
{\popandprint\algorithmicend\ \algorithmicif}%
\algdef{C}[IF]{IF}{ElsIf}[1]
{\popandprint\pushindent\algorithmicelse\ \algorithmicif\ #1\ \algorithmicthen}%
\algdef{Ce}[ELSE]{IF}{Else}{EndIf}
{\popandprint\pushindent\algorithmicelse}%
\algdef{SE}[PROCEDURE]{Procedure}{EndProcedure}[2]
{\printandpush\algorithmicprocedure\ \textproc{#1}\ifthenelse{\equal{#2}{}}{}{(#2)}}%
{\popandprint\algorithmicend\ \algorithmicprocedure}%
\algdef{SE}[FUNCTION]{Function}{EndFunction}[2]
{\printandpush\algorithmicfunction\ \textproc{#1}\ifthenelse{\equal{#2}{}}{}{(#2)}}%
{\popandprint\algorithmicend\ \algorithmicfunction}%

\makeatother

\algrenewcommand\algorithmicprocedure{\textbf{Input}}
\algrenewcommand\algorithmicreturn{\textbf{Output:}}

\begin{document}
	\bstctlcite{IEEEexample:BSTcontrol}
	
	\title{Spatiotemporal Dependable Task Execution Services in MEC-enabled Wireless Systems}
	
	\author{Mustafa~Emara,~\IEEEmembership{Student Member,~IEEE,}
		Hesham~ElSawy,~\IEEEmembership{Senior Member,~IEEE,}
		Miltiades~C.~Filippou,~\IEEEmembership{Senior Member,~IEEE,}
		Gerhard~Bauch,~\IEEEmembership{Fellow,~IEEE}
		\thanks{M. Emara and Miltiades C. Filippou are with the Germany standards R\&D team, Next Generation and Standards, Intel Deutschland GmbH, 85579 Neubiberg, Germany (e-mail: mustafa.emara, miltiadis.filippou@intel.com)}
		\thanks{H. ElSawy is with the Electrical Engineering Department, King Fahd University of Petroleum and Minerals, 31261 Dhahran, Saudi Arabia (email: hesham.elsawy@kfupm.edu.sa).}
		\thanks{G. Bauch and M. Emara are with the Institute of Communications,
			Hamburg University of Technology, Hamburg, 21073 Germany (email: bauch@tuhh.de).}}
	\maketitle
	\thispagestyle{empty}
	\maketitle
	\thispagestyle{empty}
	
	\input{miscellaneous/Abbreviations}
	\input{introduction/introduction}
	\input{system_model/system_model}
	\input{sg_analysis/sg_analysis}
	\input{qt_analysis/qt_analysis}
	\input{simulation_results/simulation_results}
	\input{conclusion/conclusion}
	\input{literature/literature}

\end{document}

%% file: miscellaneous/Abbreviations.tex
\newacronym{BS}{BS}{base station}
\newacronym{CCDF}{CCDF}{Complementary Cumulative Distribution Function}
\newacronym{C-RAN}{C-RAN}{Cloud Radio Access Network}
\newacronym{CTMC}{CTMC}{continuous time Markov chain}
\newacronym{CRA}{CRA}{computation resource availability}
\newacronym{DL}{DL}{Downlink}
\newacronym{DTMC}{DTMC}{discrete time Markov chain}
\newacronym{FDMA}{FDMA}{frequency division multiple access}
\newacronym{IoT}{IoT}{Internet of Things}
\newacronym{KPI}{KPI}{key performance indicator}
\newacronym{MEC}{MEC}{Multi-access Edge Computing}
\newacronym{M2M}{M2M}{machine-to-machine}
\newacronym{MAC}{MAC}{medium access control}
\newacronym{OSP}{OSP}{offloading success probability}
\newacronym{OoO}{OoO}{out of operation}
\newacronym{PDF}{PDF}{probability density function}
\newacronym{PPP}{PPP}{Poisson point process}
\newacronym{PM}{PM}{physical machine}
\newacronym{PMF}{PMF}{probability mass function}
\newacronym{QoS}{QoS}{quality of service}
\newacronym{TER}{TER}{task execution retainability}
\newacronym{TEC}{TEC}{task execution capacity}
\newacronym{SINR}{SINR}{signal to interference plus noise ratio}
\newacronym{TDMA}{TDMA}{time division multiple access}
\newacronym{TSP}{TSP}{transmission success probability}
\newacronym{UE}{UE}{User Equipment}
\newacronym{UL}{UL}{Uplink}
\newacronym{URLLC}{URLLC}{ultira reliable low latency communication}
\newacronym{VM}{VM}{virtual machine}

%% file: introduction/introduction.tex
\begin{abstract}

Multi-access Edge Computing (MEC) enables computation and energy-constrained devices to offload and execute their tasks on powerful servers. Due to the scarce nature of the spectral and computation resources, it is important to jointly consider i)~contention-based communications for task offloading and ii)~parallel computing and occupation of failure-prone MEC processing resources (virtual machines). The feasibility of task offloading and successful task execution with virtually no failures during the operation time needs to be investigated collectively from a combined point of view. To this end, this letter proposes a novel spatiotemporal framework that utilizes stochastic geometry and continuous time Markov chains to jointly characterize the communication and computation performance of dependable MEC-enabled wireless systems. Based on the designed framework, we evaluate the influence of various system parameters on different dependability metrics such as (i) computation resources availability, (ii) task execution retainability, and (iii) task execution capacity. Our findings showcase that there exists an optimal number of virtual machines for parallel computing at the MEC server to maximize the task execution capacity.

\end{abstract}
%
%
%
%
%
\section{Introduction}

The deployment of \ac{MEC} in 5G and beyond systems allows applications to be instantiated at the edge of the network. As a direct benefit, efficient task execution is feasible due to the \ac{MEC} servers high computation power \cite{Mach2017}. A major challenge for network operators is to provide dependable and ubiquitous computing services that meet the computing demands of devices running various heterogeneous applications (e.g., artificial intelligence, Blockchain, automotive and E-health). Efficient spectrum access for task offloading along with parallel task computation at the \ac{MEC} server are required to jointly meet such heterogeneous application requirements \cite{Porambage2018}. To ensure efficient operation, the task offloading feasibility, computation resources availability, and task execution retainability ought to be jointly quantified and optimized \cite{Bennis2018}. 

In \ac{MEC}-enabled networks, task execution at the \ac{MEC} server is strongly tied to the resources availability and the resilience to failures \cite{Bagchi2020}. In this context, various cloud-based provisioning and resilience schemes are discussed in \cite{Colman-Meixner2016}. Causes of service disruption due to \acp{PM} and \acp{VM} failures along with their analysis are provided in \cite{Birke2014}. With regard to wireless-based task offloading, \cite{Ko2018} examines the network scalability and identifies communication and computation performance frontiers. Heterogeneous networks analysis is presented in \cite{Lee2018}, where the network-wide outage probability is derived for task offloading assuming different computation architectural variants. Authors in \cite{Ko2018_b} proposed a transmission and energy efficient offloading algorithm based on a Markov decision process that accounts for the spatial and temporal network parameters. 

However, the aforementioned works either exclusively consider a dependability view of the network\cite{Bagchi2020,Colman-Meixner2016, Birke2014}, or a spatiotemporal one \cite{Ko2018, Lee2018, Ko2018_b}. As a result, the problem of feasible and dependable task execution, accounting for the joint limitation of network-wide mutual interference and parallel task computing by failure-prone \acp{VM} is still not addressed. Motivated by the above, we propose a spatiotemporal feasibility-assessment framework that entails network-wide mutual interference and temporal-based task arrivals/ processing in uplink \ac{MEC}-enabled networks. Furthermore, we adopt an individual (i.e., per-task and per-device) task execution criterion that aims to exploit the computation resources at the \ac{MEC} server if the radio conditions permit. Our analysis is then followed by the assessment of new service dependability-relevant \acp{KPI} that shed light on the system availability and task execution capability. 

%% file: system_model/system_model.tex
\vspace{-5pt}
\section{System Model}\label{sec:system_model}

\subsubsection{\textbf{Network model}} We consider a cellular uplink network, where the \acp{BS} and devices are spatially deployed in $\mathbb{R}^2$ according to two independent homogeneous Poisson point processes (PPPs), denoted by $\mathrm{\Psi}$ and $\mathrm{\Phi}$ with intensities $\lambda_b$ and $\lambda_d$, respectively. An unbounded path-loss propagation model is adopted such that the signal power attenuates at rate of $r^{-\eta}$, where $r$ is the distance and $\eta$ is the path-loss exponent. Wireless links are assumed to undergo Rayleigh fading, where the power gains of the signal of interest $h$ and the interference signal $g$,  are exponentially distributed with unit power gain. Full path-loss channel inversion power control is adopted, which implies that all devices adjust their transmit powers such that the received uplink power levels at the \ac{BS} are equal to a predetermined threshold $\rho$. 

\begin{figure}
	\begin{center}
		\includegraphics{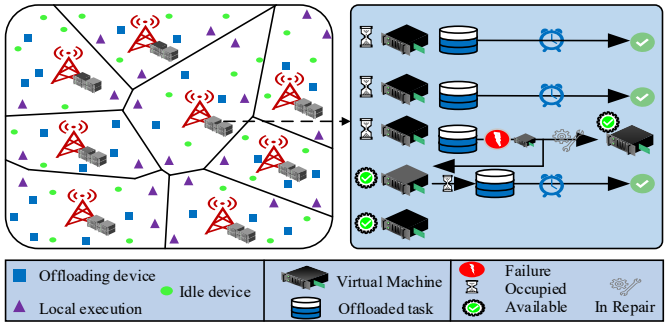}
		\caption{The considered system model with 5 VMs deployed.}
		\label{fig:system_model}
	\end{center}
\end{figure}

\subsubsection{\textbf{Offloading model}} We consider a continuous time system where task arrivals at each device are modeled via an independent Poisson process with rate $\lambda_a$ tasks/ unit time. Proactively, devices attempt to offload generated tasks by sending instructions to a MEC server collocated with the connected BS. In our system, grant-free access is assumed, where each device attempts to transmit its task instruction (i.e., not the whole task) using one of the available $C$ uplink channels randomly and uniformly without a scheduling grant from the \ac{BS} \cite{Mahmood2019}. Furthermore, let $\kappa = \frac{\lambda_d}{\lambda_b C}$ denote the average number of devices per \ac{BS} per channel and $T_s$ the transmission time of a given task's instruction. A task instruction is successfully decoded at the \ac{BS} if its received \ac{SINR} is larger than a predefined threshold, $\theta$. The \ac{OSP} of a generic device, which is denoted by $\mathcal{O}$, quantifies the probability of successful task offloading as $ \mathcal{O}= \mathbb{P}\{\text{SINR}>\theta\}$.\footnote{ACK and NACK transmission latencies are ignored as they incur negligible amount compared to $T_s$ and the task's execution time.} In the case of decoding failure (i.e., NACK is received), the device opts to compute its task locally. Accordingly, we adopt a \textit{coverage-based} offloading feasibility criterion, in which the \ac{OSP} $\mathcal{O}$ governs the offloading feasibility, thus, the offloading decision problem and its underlying parameters are not considered and left for future work. Retransmissions at the devices are not considered in the proposed model to lower the aggregate network-wide interference.

\subsubsection{\textbf{Computing model}} The \ac{MEC} server residing at each BS is equipped with a single \ac{PM} that encompasses $M_{\text{MEC}}$ \acp{VM} for parallel task computing. To account for resource sharing among the \acp{VM} (e.g., buses for I/O, CPU, memory), input/output (I/O) interference is observed within the \ac{PM} at the \ac{MEC} server. Thus, the parallel-operating \acp{VM} interfere with each other, leading to a degraded computation power \cite{Bruneo2014}. For the case of a single \ac{VM} deployment, the task's execution rate is modeled via a Poisson process with a rate of $\mu_o$ tasks/ unit time. However, to account for the I/O interference among the $M_{\text{MEC}}$ \acp{VM}, the task's execution rate of a given VM depends on the total number of \acp{VM} as follows
\vspace{-2mm}
\begin{equation}
\mu_{\text{MEC}} = \frac{\mu_o}{(1+d)^{M_{\text{MEC}}-1}},
\end{equation}
where $d$ is the computation degradation factor due to I/O interference among the $M_{\text{MEC}}$ \acp{VM} \cite{Ko2018, Bruneo2014}. For local computation of tasks, devices are assumed to be equipped with a local \ac{PM} that accommodates a single \ac{VM} (i.e., $M_{\text{loc}} = 1$, thus, no parallel processing), where the local execution rate is modeled via a Poisson process with rate $\mu_{\text{loc}}$. Moreover, a task to be computed is blocked if no VM is idle (locally or at the MEC server in case of offloading). To investigate the relative ratio between the MEC and the local computation capabilities, we define $\mu_r = {\mu_{\text{MEC}} }/{\mu_{\text{loc}}}$ which denotes the relative computation rate such that $\mu_r >> 1$.

\subsubsection{\textbf{Failure \& repair model}} Due to possible hardware and software faults, the proposed model accounts for events of \ac{VM} failures and their repairment times \cite{Fu2009,Birke2014}. The failure (repair) rate of a given \ac{VM} is modeled via a Poisson process with rate $\delta$ ($\gamma$) failure (repairment) events/ unit time.\footnote{The Poisson model is adopted in our work for task-related parameters to provide a good compromise between practical consideration of real-time events and mathematical tractability \cite{Ko2018, Lee2018}.} \acp{VM} are prone to failure regardless of being idle or occupied. A failed idle \ac{VM} is labeled as out of operation and cannot admit future tasks. Upon the failure of an occupied \ac{VM}, the \ac{PM} will handover the running task to an idle \ac{VM}, if one exists. If not, the running task is discarded and the concerned device is notified via downlink signaling. The considered system model is visualized in Fig. \ref{fig:system_model}, where one can observe a plethora of devices belonging to three categories, namely, idle, offloading and local execution devices. Focusing on a selected cell that serves a number of offloading devices, a VM fails while being in service. In this case, the task being served by this VM is transferred to an idle \ac{VM} to resume its execution. Meanwhile, based on the repair rate, the failed \ac{VM} goes back into operation to serve newly incoming tasks. 

%% file: sg_analysis/sg_analysis.tex
\vspace{-4pt}
\section{Spatial System Analysis}\label{SG_analysis}
In our work, we focus on critical applications that are sensitive to availability and reliability of the computation resources (e.g., smart agriculture, smart homes \cite{Porambage2018}), whereas latency-critical applications are left for future work. Upon task generation, the task instructions are sent to the MEC host co-located with the connected BS by uplink transmissions. Those instructions are correctly decoded, and hence the task is successfully offloaded, if the received SINR is greater than $\theta$. Otherwise, the device executes the task locally. To characterize the offloading feasibility within the network, the \ac{OSP} of a randomly selected device considering the network-wide mutual interference is
\vspace{-1mm}
\begin{align}
\mathcal{O}&=\mathbb{P}\left\{\frac{\rho h_{0}}{  \sum_{y_n \in \mathrm{\Phi}\setminus y_o} a_n P_n g_n ||y_n - z_o||^{-\eta}  +\sigma^{2}}>\theta \right\}, \\ \label{eq1}
&\myeqa\text{exp}\left \{ -\frac{\sigma^2\theta}{\rho} \right\} \mathcal{L}_{I_{\text{out}}}\left  (\frac{\theta}{\rho} \right) \mathcal{L}_{I_{\text{in}}}\left  (\frac{\theta}{\rho}\right).
\end{align} 
where $h_o$ is the channel gain between the intended device located at $y_o$ and its serving \ac{BS} located at $z_o$, $||.||$ is the Euclidean norm, $y_n$ is the $n$-th device location in the network excluding the intended device, $P_n$ is its transmit power, $g_n$ is the channel power gain between this interfering device and the intended BS, $\sigma^2$ is the noise power and $a_{n}$ equals one if the $n$-th device is transmitting on the same channel as the intended device, and zero otherwise. In addition, $(a)$ results from the exponential distribution of $h_o$ combined with the path loss inversion power control, where $\mathcal{L}_{I_{\text{out}}}(\cdot)$ and $\mathcal{L}_{I_{\text{in}}}(\cdot)$ represent the Laplace transform (LT) of the aggregate intra-cell and inter-cell interference, respectively. To provide an uplink tractable analysis, we assume that the spatial correlations between adjacent Voronoi cell areas are ignored, thus, the transmission powers of the devices are independent and identically distributed \cite{Gharbieh2017}. The aforementioned approximations are validated in Section V against independent Monte Carlo simulations. In order to quantify the total arrival rate of offloaded tasks at the \ac{MEC} server, the \ac{OSP} of each device is first calculated in the following theorem.
\begin{theorem}\label{theorem_SG}
	The \ac{OSP} for a generic device is given by
	\begin{align}\label{eq:shared}
	&\mathcal{O} \approx  \frac{\exp \left\{-\frac{\sigma^{2} \theta}{\rho}-\frac{2 \theta P_a \kappa}{(\eta-2)} {}_2F_{1}(1,1-2 / \eta, 2-2 / \eta,-\theta)\right\}}{\left(1+\frac{\theta P_a \kappa}{(1+\theta) c}\right)^{c}} \nonumber \\
	&\;\; \myeqb \; \;\frac{\exp \left\{-\frac{\sigma^{2} \theta}{\rho}- P_a \kappa \sqrt{\theta} {\arctan}\Big(\sqrt{\theta}\Big)\right\}}{\left(1+\frac{\theta P_a \kappa}{(1+\theta) c}\right)^{c}},
	\end{align}
	where $P_a=1-\text{e}^{-\left(2T_s\lambda_a\right)}$ is the device's active probability within $[-T_s,T_s]$, $\kappa = \frac{\lambda_d}{\lambda_b C}$, ${}_2F_{1}(\cdot)$ is the Gaussian hypergeometric function and $c = 3.575$. The approximation is due to the employed approximate probability distribution function (PDF) of the PPP Voronoi cell area in $\mathbb{R}^2$.
\end{theorem}
\begin{proof}
	Since full channel inversion power control with threshold $\rho$ is employed, the received power from the devices at a given \ac{BS} equals $\rho$ and the interference power from the neighboring devices is strictly lower than $\rho$. Thus, the LT of the aggregate inter-cell interference at the serving \ac{BS} is 
	\begin{equation}\label{eq:laplace_1}
	\mathcal{L}_{I_{\text {out}}}(s) \approx \exp \Big(-2 \pi P_a \lambda_d s^{\frac{2}{\eta}} \mathbb{E}\left\{[P^{\frac{2}{\eta}}\right\} \int_{(s \rho)^{\frac{-1}{\eta}}}^{\infty} \frac{y}{y^{\eta}+1} d y\Big),
	\end{equation}
	where $P_a\lambda_d$ represent the portion of active devices within the network and the approximation is due to the assumed independent transmission powers of the devices. The LT of the inter-cell interference can be evaluated as \cite[Lemma 1]{Gharbieh2017}
	\begin{align}\label{eq:laplace_2}
	\mathcal{L}_{I_{\text {in}}}(s) &\approx \mathbb{P}\{\mathcal{N}_d=0\} + \sum_{n=1}^{\infty} \frac{\mathbb{P}\{\mathcal{N}_d = n \}}{(1+s \rho)^{n}} 
	\end{align}
	where $\mathbb{E}\{\cdot\}$ is the expectation operation, $\mathcal{N}_d$ is a random variable representing the number of neighboring devices with the probability mass function $\mathbb{P}\{\mathcal{N}_d = n \} \approx \frac{\lambda_d^{n}(\lambda_b c)^{c}\mathrm{\Gamma}(n+c)}{(\lambda_d+\lambda_b c)^{n+c}\mathrm{\Gamma}(n+1) \mathrm{\Gamma}(c)}$, where $\mathrm{\Gamma}(\cdot)$ is the gamma function and $c=3.575$ is a constant defined to approximate the Voronoi cell. The theorem is proved by plugging (\ref{eq:laplace_1}) and (\ref{eq:laplace_2}) into (\ref{eq1}), followed by similar steps as done in \cite[Lemma 1]{Gharbieh2017}.
\end{proof}
Once $\mathcal{O}$ is evaluated, we can now define and evaluate the related task execution \acp{KPI} for the case of offloaded and locally executed tasks as explained in the following section. 

%% file: qt_analysis/qt_analysis.tex
\vspace{-4pt}
\section{Temporal Computational Analysis}\label{QT_analysis}

\begin{table*}[t]
	\centering
	\caption{State transitions $z=(x_I, x_O , x_F)$ of the VMs.}
	\renewcommand{\arraystretch}{1.1}
	\begin{tabular}{l l l l}
		\hline 			
		\textbf{Event} 	& \textbf{Destination state} & \textbf{Transition rate} & \textbf{Necessary condition}   \\ \hline \hline
		1- Task arrival and an idle \ac{VM} is allocated	&  $(x_I-1, x_O+1, x_F)$& $\lambda_v$ & $x_I>0$   \\		
		2-  Successful task execution at an occupied \ac{VM} &  $(x_I+1, x_O-1 , x_F)$ & $x_O\mu_v$ & $x_O > 0$   \\
		3- An idle \ac{VM} fails	&  $(x_I-1, x_O , x_F+1)$ & $x_I\delta$ & $x_I>0$   \\		
		4- An occupied \ac{VM} fails. Task is offloaded to another idle \ac{VM}	&  $(x_I-1, x_O-1, x_F+1)$& $x_O\delta$ & $x_O>0 \; \& \; x_I>0$   \\		
		5- An occupied \ac{VM} fails and task is aborted &  $(x_I, x_O-1, x_F+1)$ & $x_O\delta$ &$x_O>0 \; \& \; x_I=0$ \\		
		6- A failed \ac{VM} is repaired &  $(x_I+1, x_O , x_F-1)$ & $x_F\gamma$ & $x_F > 0$ \\
		\hline \hline
	\end{tabular}
	\label{Table:Q_states}
\end{table*}

As explained earlier, the \ac{OSP} provides an offloading feasibility assessment via controlling the aggregate load of tasks at the \ac{MEC} server. That is, the total average arrival rate of tasks to be computed at the \ac{MEC} server is $\lambda_\text{MEC} = \mathcal{O}\lambda_a \mathbb{E}\left\{ \mathcal{N}_d \right\}= \frac{\mathcal{O}\lambda_a\lambda_d}{\lambda_b}$. On the other hand, the average arrival rate of tasks to be locally computed is $\lambda_\text{loc} = \bar{\mathcal{O}}\lambda_a$ tasks/ unit time, where $\bar{\mathcal{O}} = 1-\mathcal{O}$. To analyze the temporal occupancy of the \acp{VM} either locally or at the \ac{MEC} server, we employ tools from queueing theory. To construct the proposed \ac{CTMC}, we first determine the system's state space. A general state of our model is represented by the tuple $z=(x_I, x_O , x_F)$; where $x_i; i\in\{I,O,F\}$ represents the number of \acp{VM} that are idle, occupied and failed, respectively. Let $\mathcal{S}_v=\left\{ z | \sum_j x_j = M_v; j\in\{I,O,F\} \right\}$ denote the state space, where $v\in\{\text{MEC}, \text{loc}\}$ denotes the \ac{MEC} and local systems. The steady state equations can be vectorized as $\bm{\tau}_v =[\tau_1 \; \tau_2 \; \cdots \; \tau_{\ell} \; \cdots \tau_{|\mathcal{S}_v|}]$, where $\tau_{\ell}$ is the probability of being in the $\ell$-th state. For full temporal characterization, we need to construct the state transition matrix $\bm{Q}_v$. For each system $v$, $\bm{Q}_v$ constitutes the transition rates associated with different states. To systematically construct $\bm{Q}_v$, while taking into account the different temporal events, Table \ref{Table:Q_states} is utilized, which entails the transition rates and conditions among different system states. Focusing in this work on the steady state solution, the steady state probabilities are evaluated via solving $\bm{\tau}_v\bm{Q}_v = 0,$ and $ \sum_{z\in \mathcal{S}_v} \bm{\tau}_v(z) =1.$ Let $\bm{1}$ and $\mathbf{\mathcal{I}}$ denote the all ones vector and the all ones matrix, with the appropriate sizes respectively, then, $\bm{\tau}_v$ equals $\bm{\tau}_v = \bm{1}(\bm{Q}_v + \mathbf{\mathcal{I}})^{-1}.$ Once the solution $\bm{\tau}_v$ is obtained, several dependability-based \acp{KPI} can be assessed. First, we consider the \textit{\ac{CRA}}. This metric quantifies the probability that an incoming device's task, either locally managed or offloaded to the MEC server, finds a vacant computational resource.  First, let $\mathcal{N}_{v} = \left\{z| x_I = 0, z \in \mathcal{S}_v \right\}$ denote all states with no idle \acp{VM}. Then, the \ac{CRA}, denoted as $A$, can be evaluated as 
\begin{equation}\label{avail}
A = \mathcal{O}\left(1-\sum_{ \substack{ z \in \mathcal{N}_{\text{MEC}}  } } \bm{\tau}_{\text{MEC}}(z)\right) + \bar{\mathcal{O}}\left(1-\sum_{ \substack{ z \in \mathcal{N}_\text{loc}  } } \bm{\tau}_{\text{loc}}(z)\right).
\end{equation}
Another important \ac{KPI} that quantifies the degree of successful task execution, is the \textit{\ac{TEC}}. Let $\mathcal{C}_{v} = \left\{z| x_O > 0, z \in \mathcal{S}_v\right\}$ denote all states with at least a single occupied \ac{VM}. The \ac{TEC} considers such states to evaluate the system's capability to perform task execution successfully. Denoted by $C$, the  \ac{TEC} can be computed as
\begin{equation}\label{shared} 
C =  \mathcal{O}\mu_{\text{MEC}} \!\!\! \sum_{  \substack{ z \in \mathcal{C}_{\text{MEC}}  } }\!\!\! x_O \bm{\tau}_{\text{MEC}}(z) + \bar{\mathcal{O}}\mu_{\text{loc}}\!\!\! \sum_{  \substack{ z \in \mathcal{C}_\text{loc}  } }\!\!\! x_O \bm{\tau}_{\text{loc}}(z).
\end{equation}
Finally, we consider the \textit{\ac{TER}}, which is defined as the probability that a task, once assigned to a \ac{VM}, will be computed successfully without interruption \cite{Balapuwaduge2018}. Mathematically, the \ac{TER}, denoted by $R_v$, can be evaluated as $R_v = 1-\frac{F_v}{\Lambda_v}$, where $F_v$ denotes the mean forced termination rate of ongoing tasks and $\Lambda_v$ is the effective rate in which a new task is assigned to an idle \ac{VM}. The latter can be computed similar to \eqref{avail} as $\Lambda_v=\lambda_v \left(1-\sum_{\substack{ z \in \mathcal{N}_v }} \bm{\tau}_v(z)\right)$. Moreover, let $\mathcal{F}_v = \mathcal{C}_v \cup \mathcal{N}_v$ denote all states with at least a single occupied \ac{VM} and no idle \acp{VM}. Tasks that are interrupted  in those states, because of VM failures, are dropped. Finally, $F_v$ and the \ac{TER} are evaluated as 
\begin{align}
F_v &= \delta \sum_{ \substack{ z \in \mathcal{F}_v  } } (M_v-x_F) \bm{\tau}_v(z),\\  
R &= \mathcal{O} \left( 1 - \frac{F_{\text{MEC}}}{\Lambda_{\text{MEC}}} \right) + \bar{\mathcal{O}} \left( 1 - \frac{F_{\text{loc}}}{\Lambda_{\text{loc}}} \right).
\end{align}

%% file: simulation_results/simulation_results.tex
\section{Numerical Results}\label{simulation_results}

This section aims to numerically evaluate the proposed task execution service dependability KPIs focusing on the studied MEC-enabled network. Unless otherwise stated, the list of involved network parameters are summarized in Table \ref{Table:simulation_parameters}. 

\begin{figure}
	\begin{algorithm}[H]\label{alg:iterative}
		\caption{Optimal number of deployed VMs computation.}
		\begin{algorithmic}
			\Procedure{}{$P_a, \lambda_a,\lambda_b,\lambda_d, M_{\text{MEC}}, M_{\text{loc}}, \mu_o, \mu_{\text{loc}}, d, \gamma, \delta $} 
			\LState Set $m = 1, C(0) = -\infty$, and compute $C(m)$ \Comment{$R(m)$ implies  computing $C$ in \eqref{shared} with $M_{\text{MEC}}=m$.}
			\While {$C(m)>C(m-1)$} 
			\LState Compute $C(m)$ from \eqref{shared}. 					 
			\LState Increment $m$.
			\EndWhile
			\LState \Return  $M_{\text{MEC}}^{*} = m$ and $C^*=C(M_{\text{MEC}}^{*})$.
			\EndProcedure
		\end{algorithmic}
	\end{algorithm}
\end{figure}

\begin{table}
	\centering
	\caption{Simulation parameters.}
	\renewcommand{\arraystretch}{1}
	\resizebox{0.9\columnwidth}{!}{
		\begin{tabular}{|l | l |}
			\hline 			
			\textbf{Parameter} & \textbf{value}  \\ \hline \hline
			Average number of BSs (devices) ($\lambda_b\; (\lambda_d)$) & 1 (64) BS (device)/ 10 $\text{km}^2$\\ \hline			
			Number of VMs ($M_{\text{MEC}}, M_{\text{loc}} $) & 5, 1 \\ \hline
			Number of uplink channels ($C$) & 16 \\ \hline
			Uplink power control threshold ($\rho$) & -90 dBm \\ \hline
			Path-loss exponent ($\eta$) & 4 \\ \hline
			Noise power ($\sigma^2$) & -110 dBm \\ \hline
			Detection threshold ($\theta$) & -10 dB  \\ \hline
			Task arrival rate per device ($\lambda_a$) & 0.15 tasks/ unit time \\ \hline	
			Single VM execution rate ($\mu_o$) & 3 tasks/ unit time \\ \hline	
			Local execution rate ($\mu_{\text{loc}}$) & 0.1 tasks/ unit time \\ \hline			
			VM repair rate ($\delta$) & 1 events/ unit time\\ \hline
			
			VM failure rate ($\gamma$) & 0.1 events/ unit time\\  \hline
			VM I/O degradation factor ($d$) & $0.1$ \\   \hline
	\end{tabular}}
	\label{Table:simulation_parameters}
\end{table}

Fig. \ref{fig:TSP_verf} shows the \ac{OSP} as a function of the decoding threshold $\theta$ for different device active probabilities $P_a$. The close match between the simulation and the proposed analytical framework validates the analysis and justifies the considered approximations. For increasing values of $\theta$, the \ac{OSP} decreases due to higher requirement on the link quality. For increasing values of $P_a$, the rate of task generation at the devices as well as their the probability to utilize the same uplink channel increases, thus network-wide mutual interference increases, hence, leading to lower achievable \acp{OSP}.

Focusing on the introduced \acp{KPI} in Section \ref{QT_analysis}, Fig. \ref{fig:SS_KPIs} showcases the system's performance for increasing values of $\theta$ with different system parameters. Generally, as $\theta$ increases, the \ac{OSP} decreases, thus, owing to the coverage-based offloading criterion, more devices opt to execute their tasks locally. Depending on $\mathcal{O}$, which depends on $\theta$ among other parameters, the network oscillates between an \textit{offloading-dominant} and a \textit{local execution-dominant} regime. In Fig. \ref{fig:SS_KPIs}(a), we observe that the \ac{CRA} keeps increasing till a cut-off threshold (i.e., $\theta=-6, -7$ and -8 dB for $\mu_r=20,40,80$, respectively). Operating above these threshold values, the network transitions to the local execution-dominant regime. As $\mu_r$ decreases, the CRA performance gap between the two regimes decreases, since the computational capabilities of the MEC server and device become comparable. Fig. \ref{fig:SS_KPIs}(b) presents the \ac{TER} for different per-device task arrival rates. As $\lambda_a$ increases, the contention on the radio and the computational resources increases, leading to degradation in the \ac{TER}. Fig. \ref{fig:SS_KPIs}(c) shows the \ac{TEC} for different densification ratios (i.e., average number of devices per BS per channel). In the offloading-dominant regime, high values of \ac{TEC} are achieved since the offloaded tasks leverage the computationally capable \ac{MEC} server. However, in the local execution-dominant regime, \ac{TEC} degrades till it reaches zero. We observe also the effect of $\kappa$ on the slope steepness of each curve.

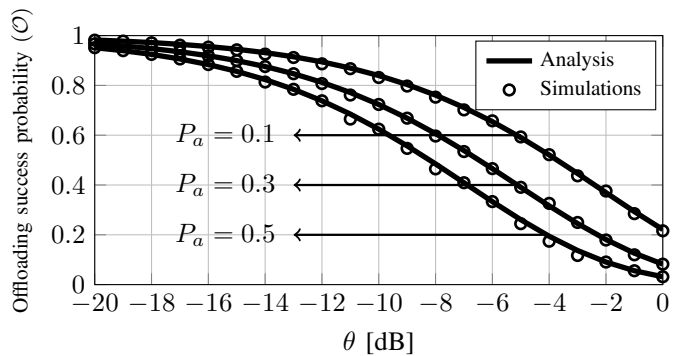
\begin{figure}
	\centering
	\input{simulation_results/figures/TSP_verf.tex}
	\caption{OSP model verification.} 
	\label{fig:TSP_verf}
\end{figure}

\begin{figure*}
	\centering
	\input{simulation_results/figures/SS_KPIs.tex}
	\caption{Steady state (a) CRA with relative computation ratios (b) TER with task arrival rates (c) TEC with densification ratios.}
	\label{fig:SS_KPIs}
\end{figure*}
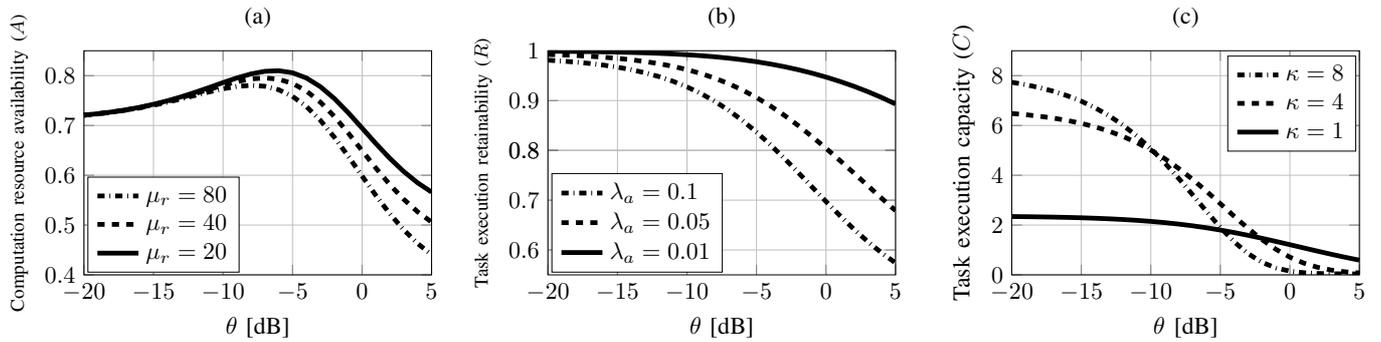

The computation resources scalability is depicted via Fig. \ref{fig:Scalable} which shows the \ac{TEC} as a function of the number of MEC server VMs $M_{\text{MEC}}$ and for three different values of the computation degradation factor $d$. The optimal number of deployed VMs for each value of $d$, calculated via Algorithm 1, which has a complexity of $O(M_{\text{MEC}})$, is shown via red circles. It is worth mentioning that the values present in Table \ref{Table:simulation_parameters} result in $P_a=0.25$ and $p = 0.83$. Thus, around 83\% of the active devices will offload their generated tasks  to the MEC server, thus, operating at the offloading-dominant regime. Nevertheless, due to the I/O interference between the employed \acp{VM} at the \ac{MEC} server, increasing $M_{\text{MEC}}$ beyond a given value, depending on the value of parameter $d$, leads to degradation in $\mu_{\text{MEC}}$ till the VM I/O interference dominates and the \ac{TEC} approaches zero. Such behavior also explains why as $d$ decreases, higher numbers of VMs are desirable. These performance results  provide network operators with important insights regarding the network dimensioning.

\begin{figure}
	\centering
	\input{simulation_results/figures/Scalable.tex}
	\caption{TEC as a function of number of VMs ($M_{\text{MEC}}$).}
	\label{fig:Scalable}
\end{figure}
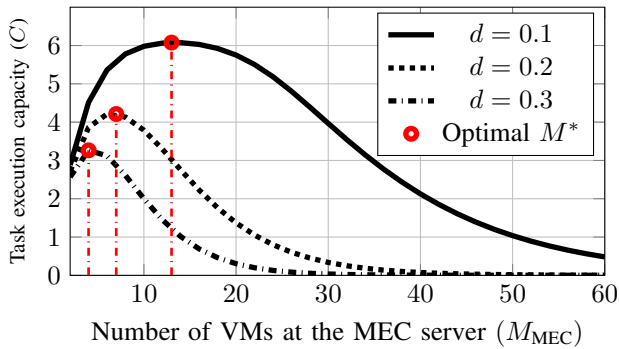

Finally, Fig. 5 shows the \ac{TER} as a function of the repair rate $\gamma$ for different values of failure rate $\delta$. For the extreme case of $\delta=0$, the \ac{TER} equals 1, independent of $\gamma$, since no \ac{VM} will ever fail. As $\delta$ increases, we observe the impact of the repair rate on the \ac{TER}, especially within the range $\gamma \in [0,1]$. For higher values of $\gamma$, the \ac{TER} starts to saturate, owing to its superiority over $\delta$, which yields it insignificant with respect to the \ac{TER}.

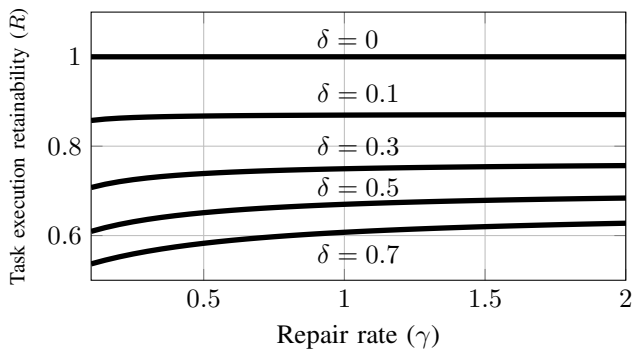
\begin{figure}
	\centering
	\input{simulation_results/figures/retain.tex}
	\caption{TER as a function of repair rate $\gamma$.}
	\label{fig:Failure}
\end{figure}


%% file: simulation_results/figures/TSP_verf.tex
%
%

%
\begin{tikzpicture}

\begin{axis}[%
width=0.85\columnwidth,
height=1.3in,
scale only axis,
xmin=-20,
xmax=0,
xlabel={$\theta\text{ [dB]}$},
ymin=0,
ymax=1,
ylabel={\footnotesize Offloading success probability $(\mathcal{O})$},
xmajorgrids,
ymajorgrids,
legend style={font = \footnotesize, legend columns=1,, at={(0.67,0.7)}, anchor=south west, legend cell align=left, align=left}
]
\addplot [color=black, line width=2.0pt]
table[row sep=crcr]{%
	-20	0.982214997224483\\
	-19	0.97767909346394\\
	-18	0.972008416114015\\
	-17	0.964931655742144\\
	-16	0.956119748425558\\
	-15	0.945177595029111\\
	-14	0.931636958096453\\
	-13	0.914952274570494\\
	-12	0.894501996009063\\
	-11	0.869599118072053\\
	-10	0.839515642031658\\
	-9	0.803526497316925\\
	-8	0.760978377850459\\
	-7	0.711387197881334\\
	-6	0.654563558441237\\
	-5	0.590758122866135\\
	-4	0.5208083580685\\
	-3	0.446256249578062\\
	-2	0.369396209155248\\
	-1	0.29320715568098\\
	0	0.221126969820591\\
};
\addlegendentry{\footnotesize Analysis}

\addplot [color=black, line width=1.0pt, draw=none,only marks, mark=o, mark options={solid, black}]
table[row sep=crcr]{%
	-20	0.981667717777389\\
	-19	0.977961795278785\\
	-18	0.971161022319486\\
	-17	0.96151435933182\\
	-16	0.953788102785126\\
	-15	0.942934634095962\\
	-14	0.925738666876877\\
	-13	0.912023906491306\\
	-12	0.886619786036807\\
	-11	0.867171978117695\\
	-10	0.831760315611401\\
	-9	0.797341025459253\\
	-8	0.752725272396118\\
	-7	0.701685845202753\\
	-6	0.65775678356949\\
	-5	0.593074724945114\\
	-4	0.52181742534209\\
	-3	0.436843857220821\\
	-2	0.375785846142525\\
	-1	0.28493810381954\\
	0	0.216093505604205\\
};
\addlegendentry{\footnotesize  Simulations}

\addplot [color=black, line width=2.0pt, forget plot]
table[row sep=crcr]{%
	-20	0.966742910548297\\
	-19	0.958368250628523\\
	-18	0.94795937900219\\
	-17	0.935064145055924\\
	-16	0.919153222096839\\
	-15	0.899620056874611\\
	-14	0.87578948019292\\
	-13	0.846939862116675\\
	-12	0.812344760508986\\
	-11	0.771340281420526\\
	-10	0.723422839029426\\
	-9	0.668377345413219\\
	-8	0.606426858987413\\
	-7	0.538381254961666\\
	-6	0.465747006543795\\
	-5	0.390748910901256\\
	-4	0.316216841075111\\
	-3	0.24531441865963\\
	-2	0.181131622165879\\
	-1	0.126216206470504\\
	0	0.0821563335472281\\
};
\addplot [color=black, line width=2.0pt, forget plot]
table[row sep=crcr]{%
	-20	0.951530331498729\\
	-19	0.939463368642236\\
	-18	0.924543304667195\\
	-17	0.906179455985228\\
	-16	0.883704907178762\\
	-15	0.856392799095033\\
	-14	0.823490457943781\\
	-13	0.784278282365348\\
	-12	0.738159879078127\\
	-11	0.684786768286958\\
	-10	0.624213371349627\\
	-9	0.557064704093202\\
	-8	0.484681017531206\\
	-7	0.40918551525906\\
	-6	0.333413666719292\\
	-5	0.260659121192642\\
	-4	0.19424028807875\\
	-3	0.13696467916073\\
	-2	0.0906336970922457\\
	-1	0.0557469907798821\\
	0	0.0315100769490762\\
};
\addplot [color=black, line width=1.0pt, draw=none,only marks, mark=o, mark options={solid, black}]
table[row sep=crcr]{%
	-20	0.966742910548297\\
	-19	0.958368250628523\\
	-18	0.94795937900219\\
	-17	0.935064145055924\\
	-16	0.919153222096839\\
	-15	0.899620056874611\\
	-14	0.87578948019292\\
	-13	0.846939862116675\\
	-12	0.807344760508986\\
	-11	0.761340281420526\\
	-10	0.723422839029426\\
	-9	0.668377345413219\\
	-8	0.5966426858987413\\
	-7	0.534381254961666\\
	-6	0.465747006543795\\
	-5	0.390748910901256\\
	-4	0.326216841075111\\
	-3	0.249531441865963\\
	-2	0.1791131622165879\\
	-1	0.1196216206470504\\
	0	0.0821563335472281\\
};
\addplot [color=black, line width=1.0pt, draw=none,only marks, mark=o, mark options={solid, black}]
table[row sep=crcr]{%
	-20	0.951530331498729\\
	-19	0.939463368642236\\
	-18	0.924543304667195\\
	-17	0.906179455985228\\
	-16	0.883704907178762\\
	-15	0.856392799095033\\
	-14	0.813490457943781\\
	-13	0.784278282365348\\
	-12	0.738159879078127\\
	-11	0.664786768286958\\
	-10	0.624213371349627\\
	-9	0.547064704093202\\
	-8	0.464681017531206\\
	-7	0.40918551525906\\
	-6	0.333413666719292\\
	-5	0.244659121192642\\
	-4	0.17424028807875\\
	-3	0.11696467916073\\
	-2	0.0906336970922457\\
	-1	0.0557469907798821\\
	0	0.0315100769490762\\
};

\draw[->, line width=0.3mm](rel axis cs:0.75,0.6) -- (rel axis cs:0.35,0.6);
\node at (rel axis cs:0.23,0.6) {$P_a = 0.1$};

\draw[->, line width=0.3mm](rel axis cs:0.74,0.4) -- (rel axis cs:0.35,0.4);
\node at (rel axis cs:0.23,0.4) {$P_a = 0.3$};

\draw[->, line width=0.3mm](rel axis cs:0.79,0.2) -- (rel axis cs:0.35,0.2);
\node at (rel axis cs:0.23,0.2) {$P_a = 0.5$};

\end{axis}
\end{tikzpicture}%

%% file: simulation_results/figures/SS_KPIs.tex
\definecolor{mycolor1}{rgb}{1.00000,0.00000,1.00000}%

\begin{tikzpicture}[scale=0.9]
\begin{groupplot}[group style={
	group name=myplot2,
	group size= 3 by 1, , horizontal sep=1.7cm}, height=1.3in,width=2in]

\nextgroupplot[title={{(a)}},
scale only axis,
xmin=-20,
xmax=5,
xlabel={$\theta\text{ [dB]}$},
ylabel={\footnotesize Computation resource availability ($A$)},
ymin=0.4,
ymax=0.85,
ytick={0.4,0.5,0.6,0.7,0.8,0.9},
xmajorgrids,
ymajorgrids,
legend style={at={(0.01,0.01)}, anchor=south west, legend cell align=left, align=left}]

\addplot [color=black, dashdotted, line width=2.0pt]
table[row sep=crcr]{%
	-20	0.720487041419672\\
	-19	0.723224893448657\\
	-18	0.726531830993501\\
	-17	0.730482879450822\\
	-16	0.735138439147541\\
	-15	0.740527298911908\\
	-14	0.746622403342652\\
	-13	0.753308694025054\\
	-12	0.760343860321288\\
	-11	0.767315423217782\\
	-10	0.773601225252982\\
	-9	0.778344700035218\\
	-8	0.780460421898842\\
	-7	0.778688296203768\\
	-6	0.771714857865798\\
	-5	0.75837437243274\\
	-4	0.737923765174601\\
	-3	0.710344240572177\\
	-2	0.676561368886553\\
	-1	0.638436704448543\\
	0	0.598445473140832\\
	1	0.55913456112031\\
	2	0.522613056169717\\
	3	0.49028931439933\\
	4	0.462873787975075\\
	5	0.440523700702284\\
};
\addlegendentry{$\mu_r = 80$}

\addplot [color=black, dashed, line width=2.0pt]
table[row sep=crcr]{%
	-20	0.720619110001032\\
	-19	0.72342986972527\\
	-18	0.726848339235296\\
	-17	0.730968548353907\\
	-16	0.735877997965285\\
	-15	0.741643097149133\\
	-14	0.748287275609108\\
	-13	0.755760379448649\\
	-12	0.763899067719331\\
	-11	0.772380212114348\\
	-10	0.780673046631695\\
	-9	0.788000954470177\\
	-8	0.793329651414187\\
	-7	0.795403670540462\\
	-6	0.792854941056955\\
	-5	0.784401415482707\\
	-4	0.769132422908712\\
	-3	0.746831901419854\\
	-2	0.71822555883726\\
	-1	0.684997707181654\\
	0	0.649486465360019\\
	1	0.614149157883539\\
	2	0.581051256538092\\
	3	0.551596049644028\\
	4	0.526517842710436\\
	5	0.506017438340578\\
};
\addlegendentry{$\mu_r = 40$}

\addplot [color=black, line width=2.0pt]
table[row sep=crcr]{%
	-20	0.720711030481192\\
	-19	0.7235729675849\\
	-18	0.727070126191864\\
	-17	0.731310431085125\\
	-16	0.736401509031033\\
	-15	0.742438258540971\\
	-14	0.749483291203361\\
	-13	0.757538414203752\\
	-12	0.766506074995109\\
	-11	0.77614159156052\\
	-10	0.786000524133889\\
	-9	0.795390957094578\\
	-8	0.803347403066049\\
	-7	0.808650096348953\\
	-6	0.809917527195646\\
	-5	0.805795539355275\\
	-4	0.795244225835465\\
	-3	0.777875101516974\\
	-2	0.754221524418492\\
	-1	0.7257817032124\\
	0	0.694736055402464\\
	1	0.663426132235007\\
	2	0.633847468904654\\
	3	0.607375589321748\\
	4	0.584751364275531\\
	5	0.566207572209453\\
};
\addlegendentry{$\mu_r = 20$}

\nextgroupplot[title={{(b)}},
scale only axis,
xmin=-20,
xmax=5,
xlabel={$\theta\text{ [dB]}$},
ymin=0.55,
ymax=1,
ytick={0.6,0.7,0.8,0.8,0.9,1},
ylabel={\footnotesize Task execution retainability ($R$)},
xmajorgrids,
ymajorgrids,legend style={at={(0.01,0.01)}, anchor=south west, legend cell align=left, align=left}]

\addplot [color=black, dashdotted, line width=2.0pt]
table[row sep=crcr]{%
	-20	0.981047213414775\\
	-19	0.979319860309441\\
	-18	0.977162663272387\\
	-17	0.974474126347269\\
	-16	0.971131841933532\\
	-15	0.966989876144855\\
	-14	0.961876808820381\\
	-13	0.955595116088286\\
	-12	0.947922862838024\\
	-11	0.938618952316667\\
	-10	0.927433363442173\\
	-9	0.914123727536145\\
	-8	0.898479027490203\\
	-7	0.880349917221463\\
	-6	0.859683079607535\\
	-5	0.836554456947032\\
	-4	0.811193933626743\\
	-3	0.783993410265761\\
	-2	0.755492353940758\\
	-1	0.726340000122088\\
	0	0.697239951624858\\
	1	0.668888260288895\\
	2	0.641917648457727\\
	3	0.616857487440612\\
	4	0.594112950103361\\
	5	0.573960402757127\\
};
\addlegendentry{$\lambda_a = 0.1$}

\addplot [color=black, dashed, line width=2.0pt]
table[row sep=crcr]{%
	-20	0.992535025242799\\
	-19	0.991581705388793\\
	-18	0.990388782714571\\
	-17	0.98889831955255\\
	-16	0.987039650260924\\
	-15	0.984727297235941\\
	-14	0.981858947835103\\
	-13	0.978313764690216\\
	-12	0.973951445985145\\
	-11	0.968612622429335\\
	-10	0.962121345291643\\
	-9	0.9542905246651\\
	-8	0.944931119186189\\
	-7	0.933865526462103\\
	-6	0.920944857835423\\
	-5	0.906068585608595\\
	-4	0.889203637204\\
	-3	0.870398899074577\\
	-2	0.849791014829665\\
	-1	0.827598890838239\\
	0	0.804107352664939\\
	1	0.779643864901996\\
	2	0.75455450041933\\
	3	0.729185190205508\\
	4	0.703871731842969\\
	5	0.678938371791693\\
};
\addlegendentry{$\lambda_a = 0.05$}

\addplot [color=black, line width=2.0pt]
table[row sep=crcr]{%
	-20	0.99903479525155\\
	-19	0.998820391990741\\
	-18	0.998551581940294\\
	-17	0.998214909306132\\
	-16	0.997793787815366\\
	-15	0.997267881941806\\
	-14	0.996612426162276\\
	-13	0.995797511224996\\
	-12	0.994787389808281\\
	-11	0.993539884804863\\
	-10	0.992006019691325\\
	-9	0.990130025417999\\
	-8	0.987849898487395\\
	-7	0.985098668748179\\
	-6	0.981806455816072\\
	-5	0.977903227646077\\
	-4	0.973321926031672\\
	-3	0.968001343132686\\
	-2	0.961887932278662\\
	-1	0.954935759964107\\
	0	0.947104154181297\\
	1	0.938353237354751\\
	2	0.928638236547721\\
	3	0.917903941264476\\
	4	0.906080725657334\\
	5	0.893083190258864\\
};
\addlegendentry{$\lambda_a = 0.01$}

\nextgroupplot[title={{(c)}},
scale only axis,
xmin=-20,
xmax=5,
xlabel={$\theta\text{ [dB]}$},
ymin=0,
ymax=9,
ylabel={Task execution capacity ($C$)},
xmajorgrids,
ymajorgrids,
legend style={at={(0.62,0.55)}, anchor=south west, legend cell align=left, align=left}]

\addplot [color=black,dashdotted, line width=2.0pt]
table[row sep=crcr]{%
	-20	7.73356331896969\\
	-19	7.63749824490398\\
	-18	7.51851637920111\\
	-17	7.37175466734704\\
	-16	7.19164879369948\\
	-15	6.97201715395201\\
	-14	6.70627110066872\\
	-13	6.38780708647321\\
	-12	6.01064063669931\\
	-11	5.57033329220733\\
	-10	5.06523416305875\\
	-9	4.4980043373952\\
	-8	3.87732252815399\\
	-7	3.21959476010831\\
	-6	2.55036056117779\\
	-5	1.90463923218624\\
	-4	1.32424903563605\\
	-3	0.848951182236804\\
	-2	0.501627913239611\\
	-1	0.277635043462287\\
	0	0.149976236140814\\
	1	0.0852702749150046\\
	2	0.0560467850783228\\
	3	0.0443722689055174\\
	4	0.0403285656937944\\
	5	0.039168655863538\\
};
\addlegendentry{$\kappa = 8$}

\addplot [color=black, dashed, line width=2.0pt]
table[row sep=crcr]{%
	-20	6.48007194154133\\
	-19	6.43088494981177\\
	-18	6.36962560135794\\
	-17	6.29353975989867\\
	-16	6.19936210939661\\
	-15	6.08328760509202\\
	-14	5.94098353267968\\
	-13	5.76767013641561\\
	-12	5.55830720650914\\
	-11	5.30793217636897\\
	-10	5.01219775036553\\
	-9	4.66814645033553\\
	-8	4.27522513541995\\
	-7	3.83647200755199\\
	-6	3.35969169024703\\
	-5	2.85827321583577\\
	-4	2.35113969958327\\
	-3	1.86126201542571\\
	-2	1.41243520549442\\
	-1	1.02478841117605\\
	0	0.710541638323519\\
	1	0.47194171453663\\
	2	0.302318371301316\\
	3	0.189410903530124\\
	4	0.11913442772712\\
	5	0.0783956125340351\\
};
\addlegendentry{$\kappa = 4$}

\addplot [color=black, line width=2.0pt]
table[row sep=crcr]{%
	-20	2.34237804459237\\
	-19	2.33619846609198\\
	-18	2.32847147230896\\
	-17	2.31882617512752\\
	-16	2.306812103579\\
	-15	2.29188731020208\\
	-14	2.27340766504464\\
	-13	2.25061942083866\\
	-12	2.22265813380463\\
	-11	2.1885581536759\\
	-10	2.1472779090476\\
	-9	2.09774663006958\\
	-8	2.038937216704\\
	-7	1.96996675262466\\
	-6	1.89021996193996\\
	-5	1.79948191687313\\
	-4	1.69805655061448\\
	-3	1.5868412228657\\
	-2	1.46732999801552\\
	-1	1.34153264682207\\
	0	1.21182010120526\\
	1	1.08073076080069\\
	2	0.950783668585764\\
	3	0.824337362191082\\
	4	0.703510794457804\\
	5	0.590156865325128\\
};
\addlegendentry{$\kappa = 1$}

\end{groupplot}

\end{tikzpicture}

%% file: simulation_results/figures/Scalable.tex
%
%

%

\definecolor{mycolor1}{rgb}{0.07451,0.62353,1.00000}%
\definecolor{mycolor2}{rgb}{0.63529,0.07843,0.18431}%

\begin{tikzpicture}

\begin{axis}[%
width=0.8\columnwidth,
height=1.4in,
scale only axis,
xmin=2,
xmax=60,
xlabel={Number of VMs at the MEC server $(M_{\text{MEC}})$},
ymin=0,
ymax=7,
ylabel={\footnotesize Task execution capacity ($C$)},
ytick={0,1,2,3,4,5,6},
yminorticks=true,
xmajorgrids,
ymajorgrids]

\addplot [color=black, line width=2.0pt]
table[row sep=crcr]{%
	2	2.88163166188199\\
	4	4.51786670483521\\
	6	5.36239576238592\\
	8	5.78025455724326\\
	10	5.97960959124516\\
	12	6.06280765067531\\
	14	6.07418966627676\\
	16	6.02814604486873\\
	18	5.92330612109804\\
	20	5.75123994458063\\
	22	5.5047256331306\\
	24	5.18554993711538\\
	26	4.8077029639912\\
	28	4.393600023417\\
	30	3.96684458235775\\
	32	3.54680119805193\\
	34	3.14679488040034\\
	36	2.774716420844\\
	38	2.43441965679353\\
	40	2.12704606579481\\
	42	1.85202597257387\\
	44	1.60776893729628\\
	46	1.39212106770124\\
	48	1.20266210169825\\
	50	1.03689564374061\\
	52	0.892368659519863\\
	54	0.766744062852011\\
	56	0.657842116102394\\
	58	0.563661105825456\\
	60	0.482384335291442\\
};
\addlegendentry{$d=0.1$}

\addplot [color=black,dotted, line width=2.0pt]
table[row sep=crcr]{%
	2	2.71601031462174\\
	4	3.86179387553576\\
	6	4.21093682822187\\
	8	4.13106513841121\\
	10	3.7884979600029\\
	12	3.29613705336094\\
	14	2.75159352342507\\
	16	2.22686773857826\\
	18	1.76170601871067\\
	20	1.37049962447581\\
	22	1.0526949826558\\
	24	0.800626935691896\\
	26	0.604138441391481\\
	28	0.452972312311818\\
	30	0.337864491978002\\
	32	0.250937536646936\\
	34	0.185739590494455\\
	36	0.137119404863226\\
	38	0.101039319068034\\
	40	0.0743784379041932\\
	42	0.0547507746129475\\
	44	0.0403482494379805\\
	46	0.0298106669534837\\
	48	0.0221210139574119\\
	50	0.0165228314501183\\
	52	0.0124559917763731\\
	54	0.00950737135487563\\
	56	0.0073733244614955\\
	58	0.00583135804526718\\
	60	0.00471889160767195\\
};
\addlegendentry{$d=0.2$}

\addplot [color=black,dashdotted, line width=2.0pt]
table[row sep=crcr]{%
	2	2.56735295610509\\
	4	3.26366551544666\\
	6	3.12064893445\\
	8	2.60861953420544\\
	10	2.00084949004406\\
	12	1.45050581164727\\
	14	1.01320253346027\\
	16	0.690060842604522\\
	18	0.461609399712053\\
	20	0.304751956312576\\
	22	0.199252084370319\\
	24	0.12937658817298\\
	26	0.0836441731160269\\
	28	0.0539969327476001\\
	30	0.034926976812149\\
	32	0.0227405964706642\\
	34	0.0149962604305036\\
	36	0.0100983428720433\\
	38	0.00701356732423774\\
	40	0.00507786751120084\\
	42	0.00386717286049914\\
	44	0.00311214323667685\\
	46	0.00264251567249365\\
	48	0.00235110038480703\\
	50	0.00217065983152259\\
	52	0.0020591532293175\\
	54	0.00199036994347945\\
	56	0.00194801112280326\\
	58	0.00192196525357187\\
	60	0.00190597272321144\\
};
\addlegendentry{$d=0.3$}

\addplot [color=red, line width=2.0pt, draw=none,only marks, mark=o, mark options={solid, red}]
table[row sep=crcr]{%
	13	6.07608549569554\\
};
\addlegendentry{Optimal $M^*$}

\addplot [color=red, line width=2.0pt, draw=none, mark=o, mark options={solid, red}]
table[row sep=crcr]{%
	7	4.2118969795718\\
};

\addplot [color=red, line width=2.0pt, draw=none, mark=o, mark options={solid, red}]
table[row sep=crcr]{%
	4	3.26366551544666\\
};

\addplot [color=red,dashdotted, line width=1.0pt]
table[row sep=crcr]{%
	13	0\\
	13	6.07608549569554\\
};

\addplot [color=red,dashdotted, line width=1.0pt]
table[row sep=crcr]{%
	7	0\\
	7	4.2118969795718\\
};

\addplot [color=red,dashdotted, line width=1.0pt]
table[row sep=crcr]{%
	4	0\\
	4	3.26366551544666\\
};

\end{axis}
\end{tikzpicture}%

%% file: simulation_results/figures/retain.tex
%
%
\begin{tikzpicture}

\begin{axis}[%
width=0.8\columnwidth,
height=1.4in,
scale only axis,
xmin=0.1,
xmax=2,
xlabel={Repair rate ($\gamma$)},
ymin=0.5,
ymax=1.1,
ylabel={\footnotesize Task execution retainability ($R$)},
xmajorgrids,
ymajorgrids]

\addplot [color=black, line width=2.0pt]
  table[row sep=crcr]{%
0.1	0.857542150962925\\
0.15	0.860691690431945\\
0.2	0.862712048702883\\
0.25	0.864115682145541\\
0.3	0.865146799786509\\
0.35	0.865936056483379\\
0.4	0.866559531798191\\
0.45	0.867064479178793\\
0.5	0.867481776525513\\
0.55	0.867832445169284\\
0.6	0.868131287814579\\
0.65	0.868389027544928\\
0.7	0.868613621510201\\
0.75	0.868811098082446\\
0.8	0.86898610751525\\
0.85	0.869142294252574\\
0.9	0.869282554791277\\
0.95	0.869409220120404\\
1	0.869524187266915\\
1.05	0.869629015768669\\
1.1	0.869724999516036\\
1.15	0.869813220997066\\
1.2	0.869894592775725\\
1.25	0.869969889575384\\
1.3	0.870039773359152\\
1.35	0.870104813127401\\
1.4	0.870165500686381\\
1.45	0.870222263312884\\
1.5	0.870275474005005\\
1.55	0.870325459839174\\
1.6	0.870372508829331\\
1.65	0.870416875592342\\
1.7	0.870458786055182\\
1.75	0.870498441387745\\
1.8	0.870536021305952\\
1.85	0.870571686859693\\
1.9	0.870605582796978\\
1.95	0.870637839577588\\
2	0.870668575095376\\
2.05	0.87069789615726\\
2.1	0.87072589975807\\
2.15	0.870752674183381\\
2.2	0.870778299966849\\
2.25	0.870802850723902\\
2.3	0.870826393880099\\
2.35	0.870848991309305\\
2.4	0.870870699894465\\
2.45	0.870891572021747\\
2.5	0.870911656017052\\
2.55	0.870930996532632\\
2.6	0.870949634890325\\
2.65	0.870967609386998\\
2.7	0.870984955566939\\
2.75	0.871001706465332\\
2.8	0.871017892826299\\
2.85	0.871033543298604\\
2.9	0.871048684611577\\
2.95	0.871063341733628\\
3	0.871077538015266\\
3.05	0.87109129531841\\
3.1	0.871104634133456\\
3.15	0.871117573685479\\
3.2	0.871130132030691\\
3.25	0.871142326144207\\
3.3	0.871154172000008\\
3.35	0.871165684643891\\
3.4	0.871176878260148\\
3.45	0.871187766232543\\
3.5	0.871198361200195\\
3.55	0.871208675108829\\
3.6	0.871218719257848\\
3.65	0.871228504343614\\
3.7	0.871238040499298\\
3.75	0.871247337331597\\
3.8	0.871256403954616\\
3.85	0.871265249021169\\
3.9	0.871273880751694\\
3.95	0.871282306961056\\
4	0.871290535083347\\
4.05	0.871298572194885\\
4.1	0.871306425035581\\
4.15	0.87131410002878\\
4.2	0.871321603299705\\
4.25	0.871328940692633\\
4.3	0.871336117786899\\
4.35	0.871343139911791\\
4.4	0.871350012160495\\
4.45	0.871356739403055\\
4.5	0.871363326298538\\
4.55	0.871369777306396\\
4.6	0.87137609669708\\
4.65	0.871382288562005\\
4.7	0.871388356822875\\
4.75	0.871394305240446\\
4.8	0.871400137422737\\
4.85	0.871405856832748\\
4.9	0.871411466795701\\
4.95	0.871416970505873\\
5	0.871422371033009\\
};

\addplot [color=black, line width=2.0pt]
  table[row sep=crcr]{%
0.1	0.609400228030472\\
0.15	0.617803308551988\\
0.2	0.624895784419067\\
0.25	0.630942548425667\\
0.3	0.636147562451094\\
0.35	0.640668265400191\\
0.4	0.644627031593553\\
0.45	0.648119833301269\\
0.5	0.651222667995159\\
0.55	0.653996305762892\\
0.6	0.656489804790807\\
0.65	0.658743130528327\\
0.7	0.660789122548646\\
0.75	0.662654984526\\
0.8	0.66436342318779\\
0.85	0.665933526809822\\
0.9	0.667381448800684\\
0.95	0.66872094415405\\
1	0.66996379387429\\
1.05	0.671120143383408\\
1.1	0.672198774338837\\
1.15	0.67320732449761\\
1.2	0.674152466740343\\
1.25	0.675040055760311\\
1.3	0.675875248976105\\
1.35	0.676662606761827\\
1.4	0.677406175978998\\
1.45	0.678109559946956\\
1.5	0.678775977337351\\
1.55	0.679408311973971\\
1.6	0.680009155126819\\
1.65	0.680580841581294\\
1.7	0.681125480520885\\
1.75	0.681644982069311\\
1.8	0.682141080184655\\
1.85	0.682615352475092\\
1.9	0.6830692374068\\
1.95	0.683504049294498\\
2	0.683920991399876\\
2.05	0.684321167409945\\
2.1	0.684705591523693\\
2.15	0.685075197339483\\
2.2	0.685430845705872\\
2.25	0.685773331673896\\
2.3	0.686103390668227\\
2.35	0.686421703977571\\
2.4	0.686728903650077\\
2.45	0.687025576867567\\
2.5	0.687312269862059\\
2.55	0.687589491429358\\
2.6	0.687857716087285\\
2.65	0.688117386919627\\
2.7	0.688368918141687\\
2.75	0.68861269741869\\
2.8	0.688849087964319\\
2.85	0.68907843044334\\
2.9	0.689301044699248\\
2.95	0.689517231325487\\
3	0.689727273096443\\
3.05	0.689931436272592\\
3.1	0.690129971792573\\
3.15	0.690323116363386\\
3.2	0.690511093458735\\
3.25	0.69069411423446\\
3.3	0.690872378368912\\
3.35	0.691046074835425\\
3.4	0.691215382613117\\
3.45	0.691380471341767\\
3.5	0.691541501925771\\
3.55	0.691698627091804\\
3.6	0.691851991904179\\
3.65	0.692001734241731\\
3.7	0.692147985239393\\
3.75	0.692290869697597\\
3.8	0.692430506462133\\
3.85	0.692567008776952\\
3.9	0.69270048461213\\
3.95	0.692831036969017\\
4	0.692958764164375\\
4.05	0.693083760095218\\
4.1	0.693206114485837\\
4.15	0.693325913118397\\
4.2	0.693443238048371\\
4.25	0.693558167805979\\
4.3	0.693670777584647\\
4.35	0.693781139417475\\
4.4	0.693889322342622\\
4.45	0.693995392558323\\
4.5	0.694099413568418\\
4.55	0.694201446318937\\
4.6	0.694301549326481\\
4.65	0.694399778798901\\
4.7	0.694496188748826\\
4.75	0.694590831100545\\
4.8	0.694683755790702\\
4.85	0.694775010863154\\
4.9	0.694864642558478\\
4.95	0.694952695398401\\
5	0.695039212265526\\
};

\addplot [color=black, line width=2.0pt]
  table[row sep=crcr]{%
0.1	0.707491199631727\\
0.15	0.714862212355881\\
0.2	0.72060733951686\\
0.25	0.725196045295544\\
0.3	0.728938294913626\\
0.35	0.732044805891755\\
0.4	0.734662917400465\\
0.45	0.736898336370104\\
0.5	0.738828667296446\\
0.55	0.740512062232333\\
0.6	0.74199289763127\\
0.65	0.743305591758119\\
0.7	0.744477229630362\\
0.75	0.745529405057644\\
0.8	0.746479537412682\\
0.85	0.747341828897345\\
0.9	0.748127971220766\\
0.95	0.748847674657374\\
1	0.749509069253268\\
1.05	0.750119012692188\\
1.1	0.750683329121359\\
1.15	0.751206996291311\\
1.2	0.751694293566916\\
1.25	0.75214892000717\\
1.3	0.752574089327299\\
1.35	0.752972606844315\\
1.4	0.753346932262957\\
1.45	0.753699231245372\\
1.5	0.754031418030171\\
1.55	0.754345190859108\\
1.6	0.754642061586392\\
1.65	0.754923380553678\\
1.7	0.755190357589697\\
1.75	0.755444079820064\\
1.8	0.755685526837958\\
1.85	0.755915583680292\\
1.9	0.756135051970655\\
1.95	0.756344659523892\\
2	0.756545068654279\\
2.05	0.756736883386786\\
2.1	0.756920655736566\\
2.15	0.757096891194106\\
2.2	0.757266053530722\\
2.25	0.75742856902062\\
2.3	0.757584830160444\\
2.35	0.757735198954739\\
2.4	0.757880009825214\\
2.45	0.758019572193134\\
2.5	0.758154172776879\\
2.55	0.758284077640633\\
2.6	0.758409534025095\\
2.65	0.758530771986793\\
2.7	0.758648005868917\\
2.75	0.758761435623581\\
2.8	0.75887124800267\\
2.85	0.758977617632281\\
2.9	0.759080707983868\\
2.95	0.759180672253392\\
3	0.759277654158593\\
3.05	0.759371788663065\\
3.1	0.759463202634912\\
3.15	0.759552015446731\\
3.2	0.759638339522977\\
3.25	0.759722280839977\\
3.3	0.759803939383351\\
3.35	0.759883409566957\\
3.4	0.759960780617144\\
3.45	0.760036136925569\\
3.5	0.760109558373584\\
3.55	0.760181120630794\\
3.6	0.760250895430186\\
3.65	0.760318950821919\\
3.7	0.760385351407707\\
3.75	0.760450158557488\\
3.8	0.760513430609936\\
3.85	0.760575223058197\\
3.9	0.760635588722105\\
3.95	0.76069457790801\\
4	0.760752238557268\\
4.05	0.760808616384285\\
4.1	0.760863755005001\\
4.15	0.760917696056539\\
4.2	0.760970479308755\\
4.25	0.761022142768324\\
4.3	0.761072722775872\\
4.35	0.761122254096805\\
4.4	0.761170770006207\\
4.45	0.76121830236829\\
4.5	0.761264881710822\\
4.55	0.761310537294862\\
4.6	0.76135529718018\\
4.65	0.761399188286632\\
4.7	0.761442236451813\\
4.75	0.761484466485237\\
4.8	0.76152590221927\\
4.85	0.761566566557069\\
4.9	0.761606481517751\\
4.95	0.761645668278828\\
5	0.761684147216385\\
};

\addplot [color=black, line width=2.0pt]
  table[row sep=crcr]{%
0.1	0.536616379804307\\
0.15	0.545012512843104\\
0.2	0.552409962362845\\
0.25	0.558957782986397\\
0.3	0.564781860041196\\
0.35	0.569987539909853\\
0.4	0.574662689244568\\
0.45	0.578880604059778\\
0.5	0.582702558626568\\
0.55	0.586179950984667\\
0.6	0.589356071546384\\
0.65	0.592267545422454\\
0.7	0.594945502914295\\
0.75	0.597416528116697\\
0.8	0.599703428353081\\
0.85	0.601825859695858\\
0.9	0.603800837094924\\
0.95	0.605643151946524\\
1	0.6073657152811\\
1.05	0.608979841013398\\
1.1	0.610495480731071\\
1.15	0.611921419153895\\
1.2	0.613265437548064\\
1.25	0.614534450923882\\
1.3	0.615734623696164\\
1.35	0.616871467578367\\
1.4	0.617949924760974\\
1.45	0.618974438851911\\
1.5	0.61994901559924\\
1.55	0.620877275050156\\
1.6	0.621762496505667\\
1.65	0.622607657392386\\
1.7	0.623415466980289\\
1.75	0.62418839571856\\
1.8	0.624928700833476\\
1.85	0.625638448727613\\
1.9	0.626319534633311\\
1.95	0.626973699902125\\
2	0.627602547253009\\
2.05	0.628207554253116\\
2.1	0.628790085263967\\
2.15	0.629351402051867\\
2.2	0.629892673232493\\
2.25	0.63041498269567\\
2.3	0.630919337135767\\
2.35	0.63140667279608\\
2.4	0.631877861520748\\
2.45	0.63233371619544\\
2.5	0.632774995647326\\
2.55	0.633202409065692\\
2.6	0.63361661999689\\
2.65	0.634018249960488\\
2.7	0.634407881727706\\
2.75	0.634786062298274\\
2.8	0.635153305607485\\
2.85	0.63551009499143\\
2.9	0.635856885435264\\
2.95	0.636194105626283\\
3	0.63652215983133\\
3.05	0.636841429615793\\
3.1	0.637152275419457\\
3.15	0.637455038002988\\
3.2	0.637750039777179\\
3.25	0.638037586025925\\
3.3	0.638317966032682\\
3.35	0.638591454119125\\
3.4	0.638858310603939\\
3.45	0.639118782688784\\
3.5	0.639373105277787\\
3.55	0.639621501736273\\
3.6	0.639864184594041\\
3.65	0.640101356197696\\
3.7	0.640333209316384\\
3.75	0.640559927704768\\
3.8	0.640781686626682\\
3.85	0.640998653342626\\
3.9	0.641210987564041\\
3.95	0.641418841876903\\
4	0.641622362137093\\
4.05	0.641821687839612\\
4.1	0.642016952463814\\
4.15	0.642208283796224\\
4.2	0.642395804232836\\
4.25	0.642579631062239\\
4.3	0.642759876731088\\
4.35	0.642936649093122\\
4.4	0.643110051642921\\
4.45	0.643280183735496\\
4.5	0.643447140792723\\
4.55	0.643611014497475\\
4.6	0.643771892976342\\
4.65	0.643929860971722\\
4.7	0.644085000003958\\
4.75	0.644237388524237\\
4.8	0.644387102058803\\
4.85	0.644534213345131\\
4.9	0.644678792460456\\
4.95	0.644820906943343\\
5	0.644960621908529\\
};

\addplot [color=black, line width=2.0pt]
  table[row sep=crcr]{%
0.1	1\\
0.15	1\\
0.2	1\\
0.25	1\\
0.3	1\\
0.35	1\\
0.4	1\\
0.45	1\\
0.5	1\\
0.55	1\\
0.6	1\\
0.65	1\\
0.7	1\\
0.75	1\\
0.8	1\\
0.85	1\\
0.9	1\\
0.95	1\\
1	1\\
1.05	1\\
1.1	1\\
1.15	1\\
1.2	1\\
1.25	1\\
1.3	1\\
1.35	1\\
1.4	1\\
1.45	1\\
1.5	1\\
1.55	1\\
1.6	1\\
1.65	1\\
1.7	1\\
1.75	1\\
1.8	1\\
1.85	1\\
1.9	1\\
1.95	1\\
2	1\\
2.05	1\\
2.1	1\\
2.15	1\\
2.2	1\\
2.25	1\\
2.3	1\\
2.35	1\\
2.4	1\\
2.45	1\\
2.5	1\\
2.55	1\\
2.6	1\\
2.65	1\\
2.7	1\\
2.75	1\\
2.8	1\\
2.85	1\\
2.9	1\\
2.95	1\\
3	1\\
3.05	1\\
3.1	1\\
3.15	1\\
3.2	1\\
3.25	1\\
3.3	1\\
3.35	1\\
3.4	1\\
3.45	1\\
3.5	1\\
3.55	1\\
3.6	1\\
3.65	1\\
3.7	1\\
3.75	1\\
3.8	1\\
3.85	1\\
3.9	1\\
3.95	1\\
4	1\\
4.05	1\\
4.1	1\\
4.15	1\\
4.2	1\\
4.25	1\\
4.3	1\\
4.35	1\\
4.4	1\\
4.45	1\\
4.5	1\\
4.55	1\\
4.6	1\\
4.65	1\\
4.7	1\\
4.75	1\\
4.8	1\\
4.85	1\\
4.9	1\\
4.95	1\\
5	1\\
};

\node at (rel axis cs:0.48,0.9) {$\delta=0$};
\node at (rel axis cs:0.5,0.7) {$\delta=0.1$};
\node at (rel axis cs:0.5,0.5) {$\delta=0.3$};
\node at (rel axis cs:0.5,0.35) {$\delta=0.5$};
\node at (rel axis cs:0.5,0.1) {$\delta=0.7$};

\end{axis}
\end{tikzpicture}%

%% file: conclusion/conclusion.tex
\vspace{-2pt}
\section{Conclusion}\label{conclusion}

This letter presents a spatiotemporal framework to characterize the network-wide task execution from a dependability perspective considering a coverage-based offloading feasibility criterion. Modeling tools are utilized to derive mathematical expressions of the OSP and a number of novel task execution dependability-based KPIs, such as CRA, TER and TEC. To yield the framework practical, \acp{VM} failures and repairment events are considered. Numerical results showcase regimes where the system transitions from the \textit{offloading-dominant} to the \textit{local execution-dominant} regime. Different system parameters such as task arrival rate, densification ratio and VM computation capabilities, are presented to obtain an understanding of the system's behavior. Finally, we show that assuming a given parameterization, there exists an optimal number of VMs, which, when deployed, maximizes the TEC.

%% file: literature/literature.tex
\bibliographystyle{./lib/IEEEtran.cls}
\bibliography{./literature/Literature_Local}